\documentclass[runningheads]{llncs}
\usepackage{amsmath}
\usepackage{cite}
\usepackage{subfigure}
\usepackage{multirow}
\usepackage{authblk}
\usepackage{indentfirst}
\usepackage{algorithm}
\usepackage{algorithmicx}
\usepackage{algpseudocode}
\usepackage{bm}
\usepackage{CJK}
\usepackage{float}
\usepackage{colortbl}
\usepackage{textcomp,booktabs}
\usepackage{array}
\usepackage{xcolor}%
\usepackage{amsfonts}%
\usepackage{amssymb}
\usepackage{graphicx}
\newcommand{\eqnb}{\begin{equation}}
\newcommand{\eqne}{\end{equation}}

\newtheorem{The}{Theorem}

\begin{document}
\title{Blockchain Queue Theory\thanks{Quan-Lin Li was supported by the National Natural Science Foundation of China under grant No. 71671158 and No. 71471160, and by the Natural Science Foundation of Hebei province under grant No. G2017203277.}}
%
%
\author{Quan-Lin Li\inst{1} \and Jing-Yu Ma\inst{1} \and Yan-Xia Chang\inst{2}}
\authorrunning{Q. L. Li et al.}
%
\institute{School of Economics and Management Sciences, Yanshan University, Qinhuangdao 066004, China\\
\email{liquanlin@tsinghua.edu.cn} \and
College of Science, Yanshan University, Qinhuangdao 066004, China}
\maketitle              
\begin{abstract}
Blockchain has many benefits including decentralization, availability, persistency, consistency, anonymity, auditability and accountability, and it also covers a wide spectrum of applications ranging from cryptocurrency, financial services, reputation system, Internet of Things, sharing economy to public and social services. Not only may blockchain be regarded as a by-product of Bitcoin cryptocurrency systems, but also it is a type of distributed ledger technologies through using a trustworthy, decentralized log of totally ordered transactions. By summarizing the literature of blockchain, it is found that more and more important research is to develop basic theory, for example, mathematical models (Markov processes, queueing theory and game models) for mining management and consensus mechanism, performance analysis and optimization of blockchain systems. In this paper, we develop queueing theory of blockchain systems and provide system performance evaluation. To do this, we design a Markovian batch-service queueing system with two different service stages, which are suitable to well express the mining process in the miners pool and the building of a new blockchain. By using the matrix-geometric solution, we obtain a system stable condition and express three key performance measures: (a) The average number of transactions in the queue, (b) the average number of transactions in a block, and (c) the average transaction-confirmation time. Finally, we use numerical examples to verify computability of our theoretical results. Although our queueing model here is simple only under exponential or Poisson assumptions, our analytic method will open a series of potentially promising research in queueing theory of blockchain systems.

\end{abstract}

\section{Introduction}

In this paper, we develop queueing theory of blockchain systems under a
dynamic behavior setting. Such a blockchain queue in general is very
necessary and useful in performance analysis and optimization of blockchain
systems, and it will also be helpful in optimal design of
blockchain technologies. To this end, we propose and analyze a Markovian
batch-service queueing system with two different service stages, which
are suitable to well express the mining process in the miners
pool and the building of a new blockchain. By using the matrix-geometric
solution, we obtain a system stable condition and express three key
measures: The average number of transactions in the queue, the average
number of transactions in a block, and the average transaction-confirmation
time. At the same time, we use some numerical examples to verify effective
computability of our theoretical results. Different from the previous works
in the literature, this paper is the first one to give a complete solution
with respect to analysis of blockchain queues. We hope that our approach
opens a new avenue to queueing analysis of more general blockchain systems
in practice, and can motivate a series of promising future research on
development of blockchain technologies.

Blockchain is one of the most popular issues discussed extensively in recent
years, and it has already changed people's lifestyle in some real areas due
to its great impact on finance, business, industry, transportation,
heathcare and so forth. Since the introduction of Bitcoin by Nakamoto \cite%
{Nak:2008}, blockchain technologies have become widely adopted in many real
applications, for example, survey work of applications by NRI \cite{Nri:2015}
and Foroglou and Tsilidou \cite{For:2015}; finance by Tsai et al. \cite%
{Tsa:2016}; business and information systems by Beck et al. \cite{Bec:2017};
applications to companies by Montemayor et al. \cite{Mon:2017}; Internet of
Things and shared economy by Huckle et al. \cite{Huc:2016}; healthcare by
Mettler \cite{Met:2016}; and the others.

So far blockchain research has obtained many important advances, readers may
refer to a book by Swan \cite{Swa:2015}; survey papers by Tschorsch and
Scheuermann \cite{Tsc:2016}, Zheng et al. \cite{Zhe:2017a} and \cite%
{Zhe:2017b}, Lin and Liao \cite{Lin:2017} and Constantinides et al. \cite%
{Con£º2018}; a key research framework by Yli-Huumo et al. \cite%
{Yli:2016}, Lindman et al. \cite{Lind:2017} and Risius and Spohrer \cite%
{Ris:2017}; consensus mechanisms by Wang et al. \cite{Wan:2018} and Debus %
\cite{Deb:2017}; blockchain economics by Catalini and Gans \cite{Cat£º2016}
and Davidson et al. \cite{Dav£º2016}; and the others by Vranken \cite%
{Vra:2017} and Dinh et al. \cite{Din:2018}.

However, little work has been done on basic theory of blockchain systems so
far, for example, developing mathematical models (e.g., Markov processes,
queueing theory, game models, optimal methods, and decision making),
providing performance analysis and optimization, and setting up useful
relations among key factors or basic parameters (e.g., block size,
transaction fee, mining reward, solving difficulty, throughput, and
efficiency).

Our blockchain queueing model focuses on analysis of the block-generation
and blockchain-building processes in the mining management, in which the sum
of the block-generation and blockchain-building times is regarded as the
transaction-confirmation time of a block. For convenience of reader's understanding, it
is necessary to simple recall some papers which discussed the miner pools
and the mining processes. A blockchain is maintained and updated by the
mining processes in which many nodes, called miners, compete for finding
answers of a very difficult puzzle-like problem. While the transactions are
grouped into a block, and then the block is pegged to the blockchain once
the key nonce is provided by means of a mining competition that an
algorithmic puzzle specialized for this block is solved. For such a mining
process, readers may refer to, for example, Bitcoin by Nakamoto \cite%
{Nak:2008}, Bhaskar and Chuen \cite{Bha:2015} and B\"{o}hme et al. \cite%
{Bhm:2015}; and blockchain by Section 2 in Zheng et al. \cite{Zhe:2017b} and
Section 2 in Dinh et al. \cite{Din:2018}. At the same time,
the mining processes are well supported by mining reward methods and consensus mechanism whose
detailed analysis was given in Debus \cite{Deb:2017} as well as an excellent
overview by Wang et al. \cite{Wan:2018}. In addition, it may be useful to
see transaction graph and transaction network by Ober et al. \cite{Obe:2013}
and Kondor et al. \cite{Kon:2014}. Finally, the mining processes were also
dicussed by game theory, e.g., see Houy \cite{Hou:2014}, Lewenberg et al. %
\cite{Lew:2015}, Kiayias et al. \cite{Kia:2016} and Biais et al. \cite%
{Bia:2018}.

Kasahara and Kawahara provided an early research (in fact, so far there have
been only their two papers in the literature) on applying queueing theory to
deal with the transaction-confirmation time for Bitcoin, in which they gave some
interesting idea and useful simulations to heuristically motivate future
promising research. See Kasahara and Kawahara \cite{Kas:2016} and Kawahara
and Kasahara \cite{Kaw:2017} for more details. In those two papers, Kasahara
and Kawahara assumed that the transaction-confirmation times follow a continuous
probability distribution function. Then they used the supplementary variable
method to set up a system of differential-difference equations (see Section
3 of Kawahara and Kasahara \cite{Kaw:2017}) by using the elapsed service
time. However, they have not correctly given the unique solution of the
system of differential-difference equations yet, although they used the
generating function technique to provide some formalized computation. For
example, the average number $E[N]$ of transactions in the system and the
average transaction-confirmation time: $E[T]=E[N]/\lambda $ given in (17) of
Kawahara and Kasahara \cite{Kaw:2017}, it is worth noting that they still
directly depend on the infinitely-many unknown numbers: $\alpha _{n}$ for $%
n=0,1,2,\ldots$. However, those unknown numbers $\alpha _{n}$ defined in
their paper are impossible to be obtained by such an ordinary technique. In
fact, we also believe that analysis of the Bitcoin queueing system with
general transaction-confirmation times, given in Kawahara and Kasahara \cite%
{Kaw:2017}, is still an interesting open problem in the future queueing
research.

To overcome the difficulties involved in Kawahara and Kasahara \cite%
{Kaw:2017}, this paper introduce two different exponential service stages
corresponding to the block-generation and blockchain-building times.
As seen in Section 2, such two service stages are very
reasonable in description of the block-generation and blockchain-building
processes. Although our blockchain queueing model is simple only under
exponential or Poisson assumptions, it is easy to see that this model is
still very interesting due to its two stages of batch services: a block of
transactions is generated and then a new blockchain is built. At the same
time, we obtain several new useful results as follows:

\textbf{(a) Stability:} This system is positive recurrent if and only if
\begin{equation*}
\frac{b\mu _{1}\mu _{2}}{\mu _{1}+\mu _{2}}>\lambda .
\end{equation*}%
Note that this stable condition is not intuitive, and it can not be obtained
by means of some simple observation. In addition, since our
transaction-confirmation time $S$ is the sum of of the block-generation time
$S_{1}$ and the blockchain-building time $S_{2}$, it is clear that our
transaction-confirmation time obeys a generalized Erlang distribution of
order $2$ with the mean $E[S]=E[S_{1}]+E[S_{2}]=1/\mu _{1}+1/\mu _{2}$. Thus
our stable condition is the same as that condition: $\lambda E[S]<b$ given
in Page 77 of Kawahara and Kasahara \cite{Kaw:2017}.

\textbf{(b) Expressions:} By using the matrix-geometric solution, we use the
rate matrix to provide simple expressions for the average number of
transactions in the queue, the average number of transactions in a block,
and the average transaction-confirmation time. At the same time, we use
numerical examples to verify computability of our theoretical results.

The structure of this paper is organized as follows. Section 2 describes an
interesting blockchain queue. Section 3 establishes a continuous-time Markov
process of GI/M/1 type, and derive a system stable condition and the
stationary probability vector by means of the matrix-geometric solution.
Section 4 provides simple expressions for three key performance measures, and uses some numerical
examples to verify computability of our theoretical results. Finally, some
concluding remarks are given in Section 5.

\section{Model Description for a Blockchain Queue}

In this section, based on the real background of blockchain, we design an
interesting blockchain queue, in which the block-generation and
blockchain-building processes are expressed as a two stages of
batch services.

When using a queueing system to model the blockchain, it is a key to set up
the service process by means of analysis of the mining management which is
related to the consensus mechanism. Here, we take a service time as the
transaction-confirmation time which is the sum of the block-generation
and blockchain-building times, that is, our service time is two stages: the
first one is generated from the mining processes, while another comes from
the network latency. In the block-generation stage, a newly generated block
is confirmed by solving a computational intensive problem by means of a
cryptographic Hash algorithm, called \textit{mining}; while a number of
nodes who compete for finding the answer is called \textit{miners}. The
winner will be awarded reward, which consists of some fixed values and fees
of transactions included in the block, and he still has the right to peg a
new block to the blockchain. In addition, a block is a list of transactions,
together with metadata including the timestamp of the current block, the
timestamp of the most previous block, and a field called a nonce which is
given by the mining winner. Therefore, it is seen that the block-generation and
blockchain-building processes, the two stages of services, can be easily
understand from the real background of blockchain, e.g., see Bitcoin
networks by Nakamoto \cite{Nak:2008} and Bhaskar and Chuen \cite{Bha:2015};
and blockchain by Section 2 in Zheng et al. \cite{Zhe:2017b} and Section 2
in Dinh et al. \cite{Din:2018}. Based on this, we can abstract the mining
processes to set up the two stages of services: the block-generation and
blockchain-building processes, so that the blockchain system is described as
a Markovian batch service queue with two different service stages, which is
depicted in Figure 1.

\begin{figure}[th]
\centering   \includegraphics[width=12.5cm]{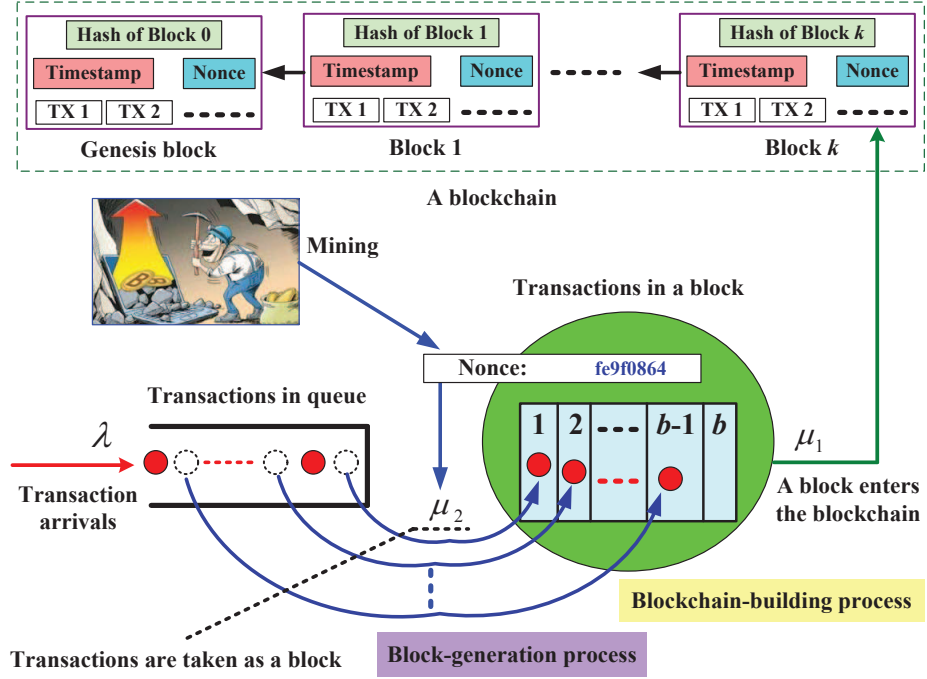}
\caption{A blockchain queueing system}
\label{figure:Fig-1}
\end{figure}

Now, from Figure 1 we provide some model descriptions as follows:

\textbf{Arrival process:} Transactions arrive at the blockchain system
according to a Poisson process with arrival rate $\lambda $. Each
transaction must first enter and queue up in an waiting room of infinite size.
Note that the arrival process of transactions is denoted in the lower left
corner of Figure 1.

\textbf{Service process:} Each arrival transaction first queues up in the
waiting room. Then it waits for being successfully mined into a block, this
is regarded as the first stage of service, called block generation. Until
the block-generation time, a group of transactions are chosen as a block and
simultaneously, a nonce is appended to the block by the mining winner. See
the lower middle part of Figure 1 for more details. Finally, the block with
the group of transactions will be pegged to a blockchain, this is regarded
as the second stage of service due to the network latency, called blockchain
building, see the right and lower part of Figure 1. In addition, the upper
part of Figure 1 also outlines the blockchain and the internal structure of
every block.

In the blockchain system, for the two stages of services, we assume that the
block-generation times in the first stage are i.i.d and exponential with
service rate $\mu _{2}$; the blockchain-building times in the second stage
are i.i.d and exponential with the service rate $\mu _{1}$.

\textbf{The block-generation discipline:} A block can consist of several
transactions but at most $b$ transactions. The transactions mined into a
block are not completely based on the First Come First Service (FCFS) with
respect to the transaction arrivals. Thus some transactions in the back
of this queue may be preferentially first chosen into the block, seen in the
lower middle part of Figure 1.

For simplification of description, our later computation will be based on
the FCFS discipline. This may be a better approximation in analysis of the
blockchain queueing system if the transactions are regarded as the
indistinctive ones under a dynamic behavior setting.

\textbf{The maximum block size:} To avoid the spam attack, the maximum block
size is limited. We assume that there are at most $b$ transactions in each
block. If the resulting block size is smaller than the maximum block size $b$%
, then the transactions newly arriving during the mining process of a block
can be accepted in the block again. If there are more than $b$ transactions
in the waiting room, then a new block will include $b$ transactions through
the full blocks to maximize the batch service ability in the blockchain
system.

\textbf{Independence:} We assume that all the random variables defined above
are independent of each other.

\textbf{Remark 1:} The arrival process of transactions at the blockchain
system may be non-Poisson (for example, Markov arrival process, renewal
process and nonhomogeneous Poisson process). On the other hand, the two
stages of batch service times may be non-exponential (for example,
phase-type distribution and general distribution). However, so far analysis of the
blockchain queues with renewal arrival process or with general service time
distribution has still been an interesting open problem in queueing research of
blockchain systems.

\textbf{Remark 2:} In the blockchain system, there are key factors, such as
the maximum block size, mining reward, transaction fee, mining strategy,
security of blockchain and so on. Based on some of them, we may develop
reward queueing models, decision queueing models, and game queueing models
in the study of blockchain systems. Analysis for these key factors will be
very necessary and useful in improving blockchain technologies in many future applications.

\section{A Markov Process of GI/M/1 Type}

In this section, for the blockchain queueing system, we establish a
continuous-time Markov process of GI/M/1 type. Based on this, we derive a
system stable condition and the stationary probability vector of this system
by means of the matrix-geometric solution.

Let $I\left( t\right) $ and$\ J\left( t\right) $ be the numbers of
transactions in the block and in the queue at time $t$, respectively. Then, $%
\left( I\left( t\right) ,J\left( t\right) \right) $ may be regarded as a
state of the blockchain queueing system at time $t$. Note that $i=0,1,\ldots
,b$ and $j=0,1,2,\ldots $, for various cases of $\left( I\left( t\right)
,J\left( t\right) \right) $ we write
\begin{align}
\mathbf{\Omega }=& \left\{ \left( i,j\right) :i=0,1,\ldots ,b,\ \ \
j=0,1,2,\ldots \right\}   \notag \\
=& \left\{ \left( 0,0\right) ,\left( 1,0\right) ,\ldots ,\left( b,0\right)
;\left( 0,1\right) ,\left( 1,1\right) ,\ldots ,\left( b,1\right) ;\ldots
;\right.   \notag \\
& \left. \left( 0,b\right) ,\left( 1,b\right) ,\ldots ,\left( b,b\right)
;\left( 0,b+1\right) ,\left( 1,b+1\right) ,\ldots ,\left( b,b+1\right)
;\ldots \right\} .  \label{eq-1}
\end{align}

\begin{figure}[th]
\centering                     \includegraphics[width=12cm]{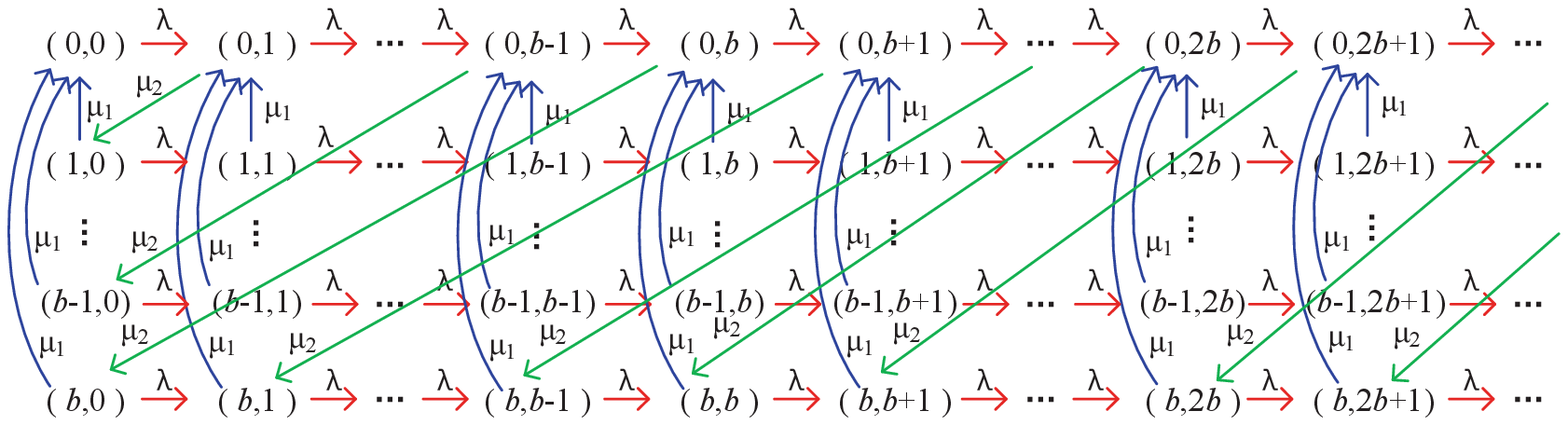}
\newline
\caption{State transition relation of a Markov process}
\label{figure:Fig-2}
\end{figure}

Let $X\left( t\right) =\left( I\left( t\right) ,J\left( t\right) \right) $.
Then $\left\{ X\left( t\right) :t\geq 0\right\} $ is a continuous-time
Markov process of GI/M/1 type on the state space $\mathbf{\Omega }$. Figure
2 denotes the state transition relation of the Markov process $\left\{
X\left( t\right) :t\geq 0\right\} $, thus its infinitesimal generator is
given by
\begin{equation}
\bm Q=\left(
\begin{array}{ccccccccc}
B_{0} & A_{0} &  &  &  &  &  &  &  \\
B_{1} & A_{1} & A_{0} &  &  &  &  &  &  \\
B_{2} &  & A_{1} & A_{0} &  &  &  &  &  \\
\vdots  &  &  & \ddots  & \ddots  &  &  &  &  \\
B_{b} &  &  &  & A_{1} & A_{0} &  &  &  \\
& A_{b} &  &  &  & A_{1} & A_{0} &  &  \\
&  & A_{b} &  &  &  & A_{1} & A_{0} &  \\
&  &  & \ddots  &  &  &  & \ddots  & \ddots
\end{array}%
\right) ,  \label{eq-2}
\end{equation}%
where
\begin{equation*}
A_{0}=\lambda I,\ \ A_{1}=\left(
\begin{array}{cccc}
-\left( \lambda +\mu _{2}\right)  &  &  &  \\
\mu _{1} & -\left( \lambda +\mu _{1}\right)  &  &  \\
\vdots  &  & \ddots  &  \\
\mu _{1} &  &  & -\left( \lambda +\mu _{1}\right)
\end{array}%
\right) ,A_{b}=\left(
\begin{array}{cccc}
0 & \cdots  & 0 & \mu _{2} \\
&  &  &  \\
&  &  &  \\
&  &  &
\end{array}%
\right) ,
\end{equation*}%
and
\begin{equation*}
B_{0}=\left(
\begin{array}{cccc}
-\lambda  &  &  &  \\
\mu _{1} & -\left( \lambda +\mu _{1}\right)  &  &  \\
\vdots  &  & \ddots  &  \\
\mu _{1} &  &  & -\left( \lambda +\mu _{1}\right)
\end{array}%
\right) ,
\end{equation*}%
\begin{equation*}
B_{1}=\left(
\begin{array}{ccccc}
0 & \mu _{2} & 0 & \cdots  & 0 \\
&  &  &  &  \\
&  &  &  &  \\
&  &  &  &
\end{array}%
\right) ,\ \ B_{2}=\left(
\begin{array}{ccccc}
0 & 0 & \mu _{2} & \cdots  & 0 \\
&  &  &  &  \\
&  &  &  &  \\
&  &  &  &
\end{array}%
\right) ,\ \ \ldots ,\ \ B_{b}=\left(
\begin{array}{cccc}
0 & \cdots  & 0 & \mu _{2} \\
&  &  &  \\
&  &  &  \\
&  &  &
\end{array}%
\right) .
\end{equation*}

For the continuous-time Markov process of GI/M/1 type, now we use the mean
drift method to discuss its system stable condition. Note that the mean
drift method for checking system stability is given a detailed introduction
in Chapter 3 of Li \cite{Li:2010} (e.g., Theorem 3.19 and the
continuous-time case in Page 172).

From Chapter 1 of Neuts \cite{Neu:1981} or Chapter 3 of Li \cite{Li:2010},
for the Markov process of GI/M/1 type, we write
\begin{equation*}
A=A_{0}+A_{1}+A_{b}=\left(
\begin{array}{ccccc}
-\mu _{2} &  &  &  & \mu _{2} \\
\mu _{1} & -\mu _{1} &  &  &  \\
\vdots  &  & \ddots  &  &  \\
\mu _{1} &  &  & -\mu _{1} &  \\
\mu _{1} &  &  &  & -\mu _{1}%
\end{array}%
\right) .
\end{equation*}%
Clearly, the Markov process $A$ with finite states is irreducible, aperiodic
and positive recurrent. In this case, we write its stationary probability
vector as $\bm\theta =\left( \theta _{0},\theta _{1},\ldots ,\theta
_{b}\right) $. Also, $\bm\theta $ is the unique solution to the system of
linear equations: $\bm\theta A=\bm0$ and $\bm\theta \bm e=1$, where $\bm e$
is a column vector of ones with proper dimension. It is easy to check that
\begin{equation*}
\bm\theta =\left( \frac{\mu _{1}}{\mu _{1}+\mu _{2}},0,\ldots ,0,\frac{\mu
_{2}}{\mu _{1}+\mu _{2}}\right).
\end{equation*}

The following theorem provides a necessary and sufficient condition under
which the Markov process $\bm Q$ is positive recurrent.

\begin{The}
The Markov process $\bm Q$ of GI/M/1 type is positive recurrent if and only
if
\begin{equation*}
\frac{b\mu _{1}\mu _{2}}{\mu _{1}+\mu _{2}}>\lambda .
\end{equation*}
\end{The}

\begin{proof}
By using the mean drift method given in Chapter 3 of Li \cite{Li:2010} (e.g., Theorem 3.19 and the
continuous-time case in Page 172), it is easy to see that the Markov process $\bm Q$ of GI/M/1
type is positive recurrent if and only if
\begin{equation}
\bm  \theta A_{0}\bm  e< b\bm  \theta A_{b}\bm  e. \label{eq-3}%
\end{equation}
Note that
\begin{equation}
\bm  \theta A_{0}\bm  e=\lambda, \label{eq-4}%
\end{equation}
and
\begin{equation}
b \bm  \theta A_{b}\bm  e= b \theta_{0}\mu_{2} = \frac{b \mu_{1}\mu_{2}}{\mu_{1}+\mu_{2}}, \label{eq-5}%
\end{equation}
thus we obtain
\[
\frac{b \mu_{1}\mu_{2}}{\mu_{1}+\mu_{2}}>\lambda.
\]
This completes the proof. \textbf{{\rule{0.08in}{0.08in}}}
\end{proof}

When the Markov process $\bm Q$ of GI/M/1 type is positive recurrent, we
write its stationary probability vector as
\begin{equation*}
\bm\pi =\left( \bm\pi _{0},\bm\pi _{1},\bm\pi _{2},\ldots \right) ,
\end{equation*}%
where
\begin{align*}
\bm\pi _{0}& =\left( \pi _{0,0},\pi _{1,0},\ldots ,\pi _{b,0}\right) , \\
\bm\pi _{1}& =\left( \pi _{0,1},\pi _{1,1},\ldots ,\pi _{b,1}\right) , \\
\bm\pi _{k}& =\left( \pi _{0,k},\pi _{1,k},\ldots ,\pi _{b,k}\right) ,\ \ \
k\geq 2.
\end{align*}

For the Markov process $\bm Q$ of GI/M/1 type, to compute its stationary
probability vector, we need to first obtain the rate matrix $R$, which is
the minimal nonnegative solution to the following nonlinear matrix equation
\begin{equation}
R^{b}A_{b}+RA_{1}+A_{0}=\mathbf{0}.  \label{eq-6}
\end{equation}

In general, it is very complicated to provide expression for the unique solution to this nonlinear matrix equation
due to the term $R^{b}A_{b}$ of size $b+1$. In fact, for the blockchain
queueing system, here we can not also provide an explicit expression for the
rate matrix $R$. For example, we consider a special case with $b=2$. For this
simple case, we have
\begin{equation*}
A_{0}=\lambda I,A_{1}=\left(
\begin{array}{ccc}
-\left( \lambda +\mu _{2}\right)  &  &  \\
\mu _{1} & -\left( \lambda +\mu _{1}\right)  &  \\
\mu _{1} &  & -\left( \lambda +\mu _{1}\right)
\end{array}%
\right) ,A_{2}=\left(
\begin{array}{ccc}
0 & 0 & \mu _{2} \\
&  &  \\
&  &
\end{array}%
\right) .
\end{equation*}%
From computing $R^{2}A_{2}+RA_{1}+A_{0}=\mathbf{0}$, it is easy to see that
for the simple case with the maximal block size $b=2$, we can not provide an
explicit expression for the rate matrix $R$ of size $3$ yet.

Although the rate matrix $R$ has not an explicit expression, we can use some
iterative algorithms, given in Chapters 1 and 2 of Neuts \cite{Neu:1981}, to give
its numerical solution. Here, an effective iterative algorithm is described
as
\begin{align*}
R_{0}& =\mathbf{0,} \\
R_{N+1}& =\left( R_{N}^{b}A_{b}+A_{0}\right) \left( -A_{1}\right) ^{-1}.
\end{align*}%
Note that this algorithm is fast convergent, that is, after a finite number
of iterative steps, we can numerically given a solution of higher precision
for the rate matrix $R$.

The following theorem directly comes from Theorem 1.2.1 of Chapter 1 in
Neuts \cite{Neu:1981}. Here, we restate it without a proof.

\begin{The}
If the Markov process $\bm Q$ of GI/M/1 type is positive recurrent, then the
stationary probability vector $\bm \pi=\left( \bm \pi_{0},\bm \pi_{1}, \bm %
\pi_{2},\ldots\right)$ is given by
\begin{equation}
\bm \pi_{k}=\bm \pi_{0}R^{k}, \ \ \ k\geq 1,  \label{eq-8}
\end{equation}
where the vector $\bm \pi_{0}$ is positive, and it is the unique solution to
the following system of linear equations:
\begin{align}
& \left. \bm \pi_{0}B\left[ R\right] =\bm \pi_{0},\right.  \label{eq-7} \\
& \left. \bm \pi_{0}\left( I-R\right) ^{-1}\bm e=1,\right.  \notag
\end{align}
and the matrix $B\left[ R\right] =\sum\nolimits_{k=0}^{b}R^{k}B_{k}$ is
irreducible and stochastic.
\end{The}

\section{Performance Analysis}

In this section, we provide performance analysis of the blockchain queueing
system. To this end, we provide three key performance measures and give
their simple expressions by means of the vector $\bm \pi_{0}$ and the rate
matrix $R$. Finally, we use numerical examples to verify computability of
our theoretical results, and show how the performance measures depend on
some key parameters of this system.

When the blockchain queueing system is stable, we write
\begin{equation*}
I=lim_{t\rightarrow +\infty }I\left( t\right) ,\ \ \ J=lim_{t\rightarrow
+\infty }J\left( t\right) .
\end{equation*}

\textbf{(a) The average number of transactions in the queue}

It follows from (\ref{eq-7}) and (\ref{eq-8}) that
\begin{equation*}
E\left[ J \right] =\overset{\infty}{\underset{j=0}{\sum}}j\overset{b}{%
\underset{i=0}{\sum}}\pi_{i,j}=\overset{\infty}{\underset{j=0}{\sum}}j\bm %
\pi_{j}\bm e=\bm \pi _{0}R\left( I-R\right) ^{-2}\bm e.
\end{equation*}

\textbf{(b)} \textbf{The average number of transactions in the block}

Let $\bm           h=\left( 0,1,2,\ldots,b\right) ^{T}$. Then
\begin{equation*}
E\left[ I \right] =\overset{\infty}{\underset{j=0}{\sum}}\underset{i=0}{%
\overset{b}{\sum}}i\pi_{i,j}=\underset{j=0}{\overset{\infty }{\sum}}\bm %
\pi_{j}\bm h=\bm \pi_{0}\left( I-R\right) ^{-1}\bm h.
\end{equation*}

\textbf{(c) The average transaction-confirmation time}

In the blockchain system, the transaction-confirmation time is the time
interval from the time epoch that a transaction arrives at the waiting room
to the point that the blockchain pegs a block with this transaction. In fact, the
transaction-confirmation time is the sojourn time of the transaction in the
blockchain system, and it is also the sum of the block-generation and
blockchain-building times of a block with this transaction.

Let $T$ denote the transaction-confirmation time of any transaction when the
blockchain system is stable.

The following theorem provides expression for the average
transaction-confirmation time by means of the stationary probability vector.

\begin{The}
If the blockchain queueing system is stable, then the average
transaction-confirmation time $E[T]$ is given by%
\begin{equation*}
E[T]= \overset{\infty}{\underset{k=0}{\sum}} \overset{b-1}{\underset{l=0}{%
\sum}} \pi_{0,kb+l}\left( k+1\right) \left( \frac{1}{\mu_{1}}+\frac{1}{%
\mu_{2}}\right) + \overset{b}{\underset{i=1}{\sum}} \overset{\infty}{%
\underset{k=0}{\sum}} \overset{b-1}{\underset{l=0}{\sum}} \pi_{i,kb+l}\left[
\frac{1}{\mu_{1}}+\left( k+1\right) \left( \frac{1}{\mu_{1}}+\frac{1}{\mu_{2}%
}\right) \right] .
\end{equation*}
\end{The}

\begin{proof}
It is clear that $\left(  i,kb+l\right)  $ is a state of the blockchain
system for $i=0,1,\ldots,b,k=0,1,\ldots,$ and $l=0,1,\ldots,b-1.$

When a transaction arrives at the blockchain system at time $t$, it
observes and finds that there are $i$ transactions in the block and
$kb+l$ transactions in the queue. Based on the two stages of exponential
service times and by using the stationary probability vector $\bm \pi$, we apply the law of total
probability to be able to compute the average transaction-confirmation time.
For this, our computation will have two different cases as follows:

{\bf Case one: $i=0$}. In this case, the transaction finds that there is no transaction in the block
at the beginning moment, thus its transaction-confirmation time includes
$k+1$ block-generation times $(k+1)/\mu_{2}$ and $k+1$ blockchain-building times $(k+1)/\mu_{1}$. Thus we obtain
\[
E[T_{|i=0}]=\overset{\infty}{\underset{k=0}{\sum}} \overset{b-1}{\underset{l=0}{%
\sum}} \pi_{0,kb+l}\left(  k+1\right)  \left(  \frac{1}{\mu_{1}}%
+\frac{1}{\mu_{2}}\right).
\]

{\bf Case two: $i=1,2,\ldots,b$}. In this case, the transaction finds that there are
$i$ transactions in the block, thus its transaction-confirmation time includes
$k+1$ block-generation times $(k+1)/\mu_{2}$ and $k+2$ blockchain-building times $(k+2)/\mu_{1}$. We obtain
\[
E[T_{|i\neq 0}]=\overset{b}{\underset{i=1}{\sum}} \overset{\infty}{%
\underset{k=0}{\sum}} \overset{b-1}{\underset{l=0}{\sum}} \pi_{i,kb+l}\left[  \frac{1}{\mu
_{1}}+\left(  k+1\right)  \left(  \frac{1}{\mu_{1}}+\frac{1}{\mu_{2}}\right)
\right]  .
\]
Therefore, by means of considering all the different cases, the average transaction-confirmation time $E[T]$ is given by
\[
E[T]= \overset{\infty}{\underset{k=0}{\sum}}
\overset{b-1}{\underset{l=0}{\sum}}
\pi_{0,kb+l}\left(  k+1\right)  \left(  \frac{1}{\mu_{1}}+\frac{1}%
{\mu_{2}}\right)  +
\overset{b}{\underset{i=1}{\sum}}
\overset{\infty}{\underset{k=0}{\sum}}
\overset{b-1}{\underset{l=0}{\sum}}
\pi_{i,kb+l}\left[
\frac{1}{\mu_{1}}+\left(  k+1\right)  \left(  \frac{1}{\mu_{1}}+\frac{1}%
{\mu_{2}}\right)  \right]  .
\]
This completes the proof. \textbf{{\rule{0.08in}{0.08in}}}
\end{proof}

Let $\langle \mathbf{x} \rangle_{|_{i=0}}$ be the 1st element of the vector $%
\mathbf{x}$. The following theorem provides a simple expression for the
average transaction-confirmation time $E[T]$ by means of the vector $\bm %
\pi_{0}$ and the rate matrix $R$.

\begin{The}
If the blockchain queueing system is stable, then the average
transaction-confirmation time $E[T]$ is given by
\begin{eqnarray*}
E\left[ T\right] &=&\frac{1}{\mu _{1}}\left[ \bm \pi_{0} \left( I-R\right)
^{-1} \bm e -\left\langle \bm \pi_{0}\left( I-R\right) ^{-1}\right\rangle
_{|_{i=0}}\right] \\
&&+\left( \frac{1}{\mu _{1}}+\frac{1}{\mu _{2}}\right) \bm \pi_{0}\left(
I-R^{b}\right) ^{-1}\left( I-R\right) ^{-1} \bm e.
\end{eqnarray*}
\end{The}

\begin{proof}
By using Theorem 3, we give some corresponding computation. Note that
\[
\frac{1}{\mu _{1}}\overset{b}{\underset{i=1}{\sum }}\overset{\infty }{%
\underset{k=0}{\sum }}\overset{b-1}{\underset{l=0}{\sum }}\pi _{i,kb+l}=%
\frac{1}{\mu _{1}}\left[ \overset{\infty }{\underset{k=0}{\sum }}\overset{b-1%
}{\underset{l=0}{\sum }}\bm \pi _{kb+l}\bm e-\overset{\infty }{\underset{k=0}{\sum }}%
\overset{b-1}{\underset{l=0}{\sum }}\pi _{0,kb+l}\right] ,
\]
\[
\overset{\infty }{\underset{k=0}{\sum }}\overset{b-1}{\underset{l=0}{\sum }}%
\bm \pi _{kb+l}\bm e=\overset{\infty }{\underset{k=0}{\sum }}\overset{b-1}{\underset{%
l=0}{\sum }}\bm \pi _{0}R^{kb+l}\bm e=\bm \pi _{0}\left( I-R\right) ^{-1}\bm e,
\]
\[
\overset{\infty }{\underset{k=0}{\sum }}\overset{b-1}{\underset{l=0}{\sum }}%
\pi _{0,kb+l}=\left\langle \overset{\infty }{\underset{k=0}{\sum }}\overset{%
b-1}{\underset{l=0}{\sum }}\bm \pi _{kb+l}\right\rangle _{|_{i=0}}=\left\langle
\bm \pi _{0}\left( I-R\right) ^{-1}\right\rangle _{|_{i=0}},
\]
we obtain
\[
\frac{1}{\mu _{1}}\overset{b}{\underset{i=1}{\sum }}\overset{\infty }{%
\underset{k=0}{\sum }}\overset{b-1}{\underset{l=0}{\sum }}\pi _{i,kb+l}=%
\frac{1}{\mu _{1}}\left[ \bm \pi _{0}\left( I-R\right) ^{-1}\bm e-\left\langle \bm \pi
_{0}\left( I-R\right) ^{-1}\right\rangle _{|_{i=0}}\right] .
\]
On the other hand, since
\[
\overset{b}{\underset{i=0}{\sum }}\overset{\infty }{\underset{k=0}{\sum }}%
\overset{b-1}{\underset{l=0}{\sum }}\pi _{i,kb+l}\left( k+1\right) \left(
\frac{1}{\mu _{1}}+\frac{1}{\mu _{2}}\right) =\left( \frac{1}{\mu _{1}}+%
\frac{1}{\mu _{2}}\right) \overset{\infty }{\underset{k=0}{\sum }}\overset{%
b-1}{\underset{l=0}{\sum }}\left( k+1\right) \bm \pi _{kb+l} \bm e,
\]
we get
\begin{eqnarray*}
& &\overset{\infty }{\underset{k=0}{\sum }}\overset{b-1}{\underset{l=0}{\sum }}%
\left( k+1\right) \bm \pi _{kb+l} = \overset{\infty }{\underset{k=0}{\sum }}%
\overset{b-1}{\underset{l=0}{\sum }}\left( k+1\right) \bm \pi _{0}R^{kb+l} \\
&=&\bm \pi _{0}\overset{\infty }{\underset{k=0}{\sum }}\left( k+1\right)
R^{kb}\left( I-R^{b}\right) \left( I-R\right) ^{-1} \\
&=&\bm \pi _{0}\left( I-R^{b}\right) ^{-1}\left( I-R\right) ^{-1}.
\end{eqnarray*}%
This leads our desired result. The proof is completed.  \textbf{{\rule{0.08in}{0.08in}}}
\end{proof}

In the remainder of this section, we provide some numerical examples to
verify computability of our theoretical results and to analyze how the three
performance measures $E\left[ J\right]$, $E\left[ I\right]$ and $E\left[ T%
\right]$ depend on some crucial parameters of the blockchain queueing system.

In the numerical examples, we take some common parameters: Block-generation
service rate $\mu_{1}\in\left[ 0.05,1.5\right] $, blockchain-building service
rate $\mu_{2}=2$, arrival rate $\lambda=0.3$, maximum block size
$b_{1}=40, b_{2}=80, b_{3}=320$.

From the left part of Figure 3, it is seen that $E\left[ J\right] $ and $E%
\left[ I\right] $ decrease, as $\mu_{1}$ increases. At the same time, from
the right part of Figure 3, $E\left[ J\right] $ decreases as $b$ increases,
but $E\left[ I\right] $ increases as $b$ increases.

\begin{figure}[th]
\centering        \includegraphics[width=6cm]{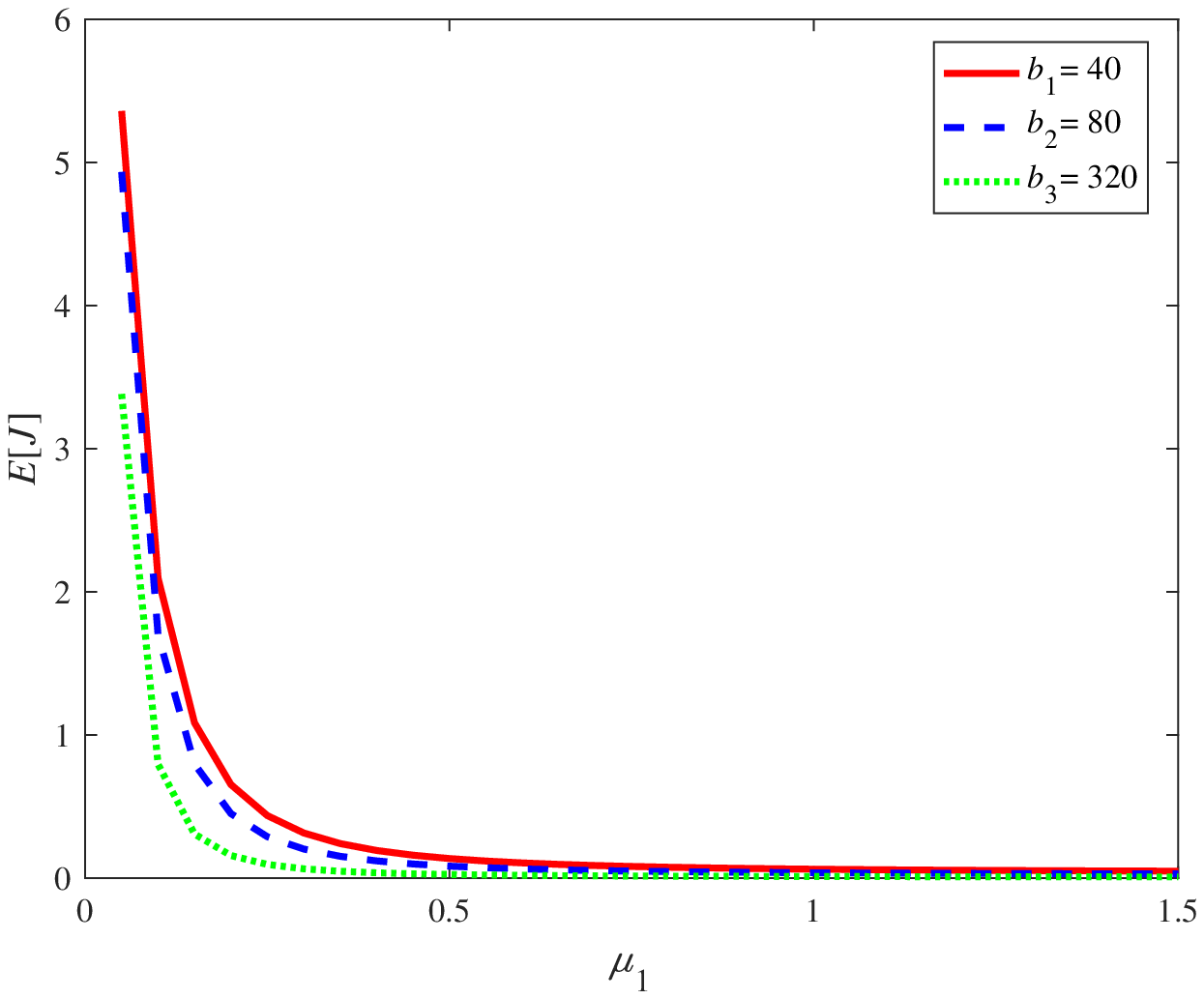} \centering        %
\includegraphics[width=6cm]{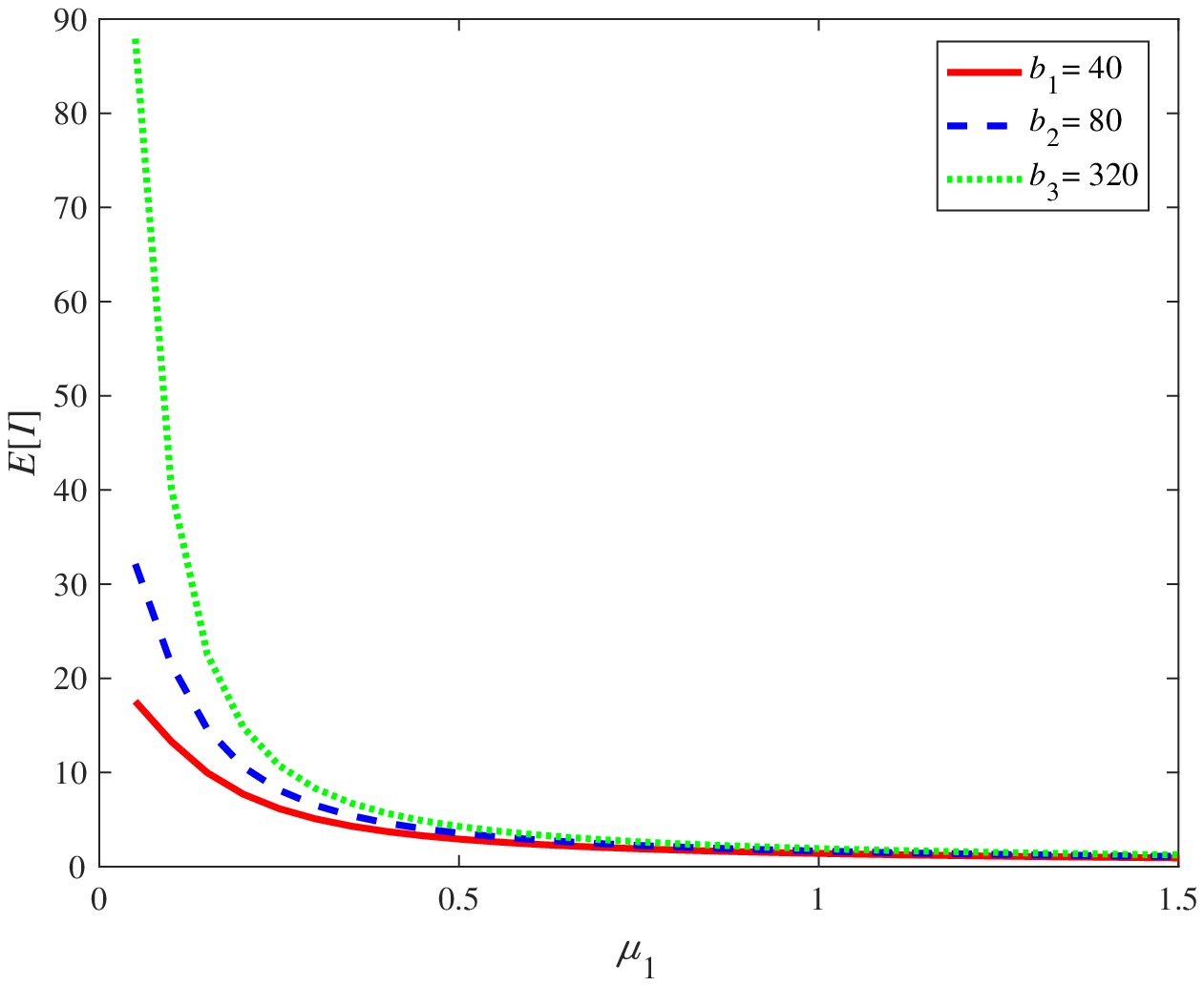}
\caption{$E\left[ J\right] $, $E\left[ I\right] $ vs. $\protect\mu_{1}$ for
three different $b_{1},b_{2},b_{3}$.}
\label{figure:Fig-3}
\end{figure}

From Figure 4, it is seen that $E\left[ T\right] $ decreases, as $\mu_{1}$
increases; while it decreases, as $b$ increases.

\begin{figure}[th]
\centering        \includegraphics[width=11cm]{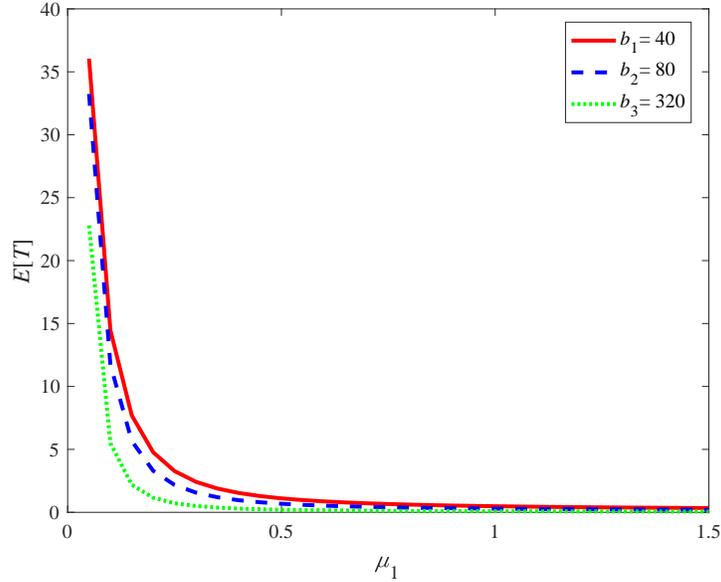}
\caption{$E\left[ T\right] $ vs. $\protect\mu_{1}$ for three different $%
b_{1},b_{2},b_{3}$.}
\label{figure:fig-4}
\end{figure}

Finally, we specifically observe the impact of the maximal block size $b$ on the $E\left[ T\right] $.
Comparing with the parameters given in Figures 3 and 4, we only change the two parameters $b$ and $\lambda$:
$b=40, 41, 42, \ldots, 500$, $\lambda_1=0.1, \lambda_2=0.9, \lambda_3=1.5$.
From Figure 5, it is seen that $E\left[ T\right] $ decreases, as $b$
increases. In addition, it is observed that there exists a critical value $%
\eta$ such that when $b \leq \eta$, $E\left[ T\right] $ increases, as $\lambda$
increases. On the contrary, when $b > \eta$, $E\left[ T\right] $ increases, as $\lambda$
decreases. In fact, such a difference is also intuitive that the block generation time becomes
bigger as $\lambda$ decreases.

\begin{figure}[th]
\centering        \includegraphics[width=11cm]{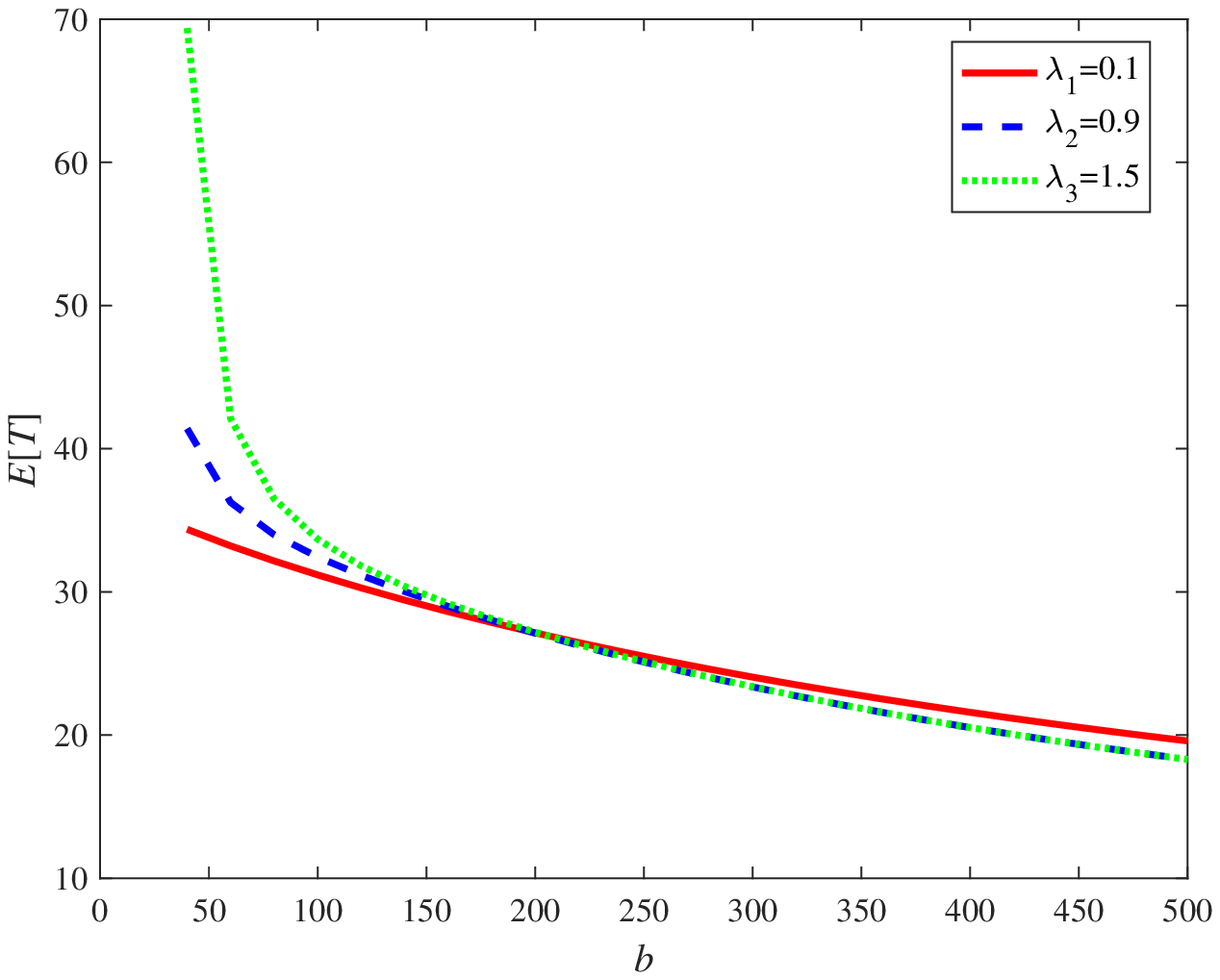}
\caption{$E\left[ T\right] $ vs. $b$ for three different $%
\lambda_{1}, \lambda_{2}, \lambda_{3}$.}
\label{figure:fig-5}
\end{figure}

\section{Concluding Remarks}

In this paper, our aim is to focus on how to develop queueing theory of
blockchain systems. To do this, we design an interesting Markovian
batch-service queueing system with two different service stages, while the
two stages can well express the mining process in the miners pool and the
building of a new blockchain. By using the matrix-geometric solution, we not
only obtain a system stable condition but also simply express three key performance
measures: The average number of transactions in the queue, the average
number of transactions in a block, and the average transaction-confirmation
time. Finally, we use numerical examples to verify computability of our
theoretical results. Along these lines, we will continue our future research
on the following directions:

-- Considering more general blockchain queueing systems, for example, the
Markov arrival process of transactions, and two stages of phase-type batch
services.

-- Analyzing multiple classes of transactions in the blockchain systems.
Also, the transactions are dealt with in the block-generation and
blockchain-building processes according to a priority service discipline.

-- When the arrivals of transactions are a renewal process, and the
block-generation times or the blockchain-building times follow general
probability distributions, an interesting research work is to focus on fluid
and diffusion approximations of the blockchain systems.

-- Setting up reward function with respect to cost structures, transaction
fees, mining reward, consensus mechanism, security and so forth.
Our aim is to find optimal policies in the blockchain systems.

-- Further developing stochastic optimization and control, Markov decision
processes and stochastic game theory in the blockchain systems.

\end{document}